\documentclass[a4paper]{article}

\usepackage[utf8]{inputenc}
\usepackage{hyperref}
\usepackage{bm, mathtools}
\usepackage{amsmath, amsthm, amsfonts, amssymb}
\usepackage{appendix}
\usepackage{graphicx}
\usepackage{subcaption}
\usepackage[round]{natbib}
\usepackage[a4paper, total={6in, 8in}]{geometry}

\newcommand\given[1][]{\:#1\vert\:}

\newcommand\norm[1]{\left|\left|#1\right|\right|}

\newcommand{\e}{\mathrm{e}^}
\newcommand{\Kdim}{\mathrm{K}}

\newcommand{\PoARX}{\mathrm{PoARX}}
\newcommand{\PoARXpq}[3][2]{\ensuremath{\PoARX_{#1}(#2,#3)}}

\newcommand{\pgfP}[2][]{P_{#1}(#2)}

\newcommand{\infoF}[1][T]{\mathcal{F}_{#1}}
\newcommand{\infoG}[1][T]{\mathcal{G}_{#1}}

\newtheorem{theorem}{Theorem}
\newtheorem{definition}{Definition}
\newtheorem{lemma}{Lemma}

\DeclareMathOperator{\E}{E}

\DeclareMathOperator{\Cov}{Cov}
\DeclareMathOperator{\vect}{vec}
\DeclareMathOperator{\interior}{int}
\DeclareMathOperator*{\argmax}{arg\,max}

\begin{document}

\title{PoARX Modelling for Multivariate Count Time Series}

\author{Jamie Halliday\footnote{University of Manchester, United Kingdom, \texttt{jamie.halliday@manchester.ac.uk}} \and 
Georgi N. Boshnakov\footnotemark[\value{footnote}]}

\maketitle


\abstract{
This paper introduces multivariate Poisson autoregressive models with exogenous
covariates (PoARX) for modelling multivariate time series of counts.  We obtain
conditions for the PoARX process to be stationary and ergodic before proposing a
computationally efficient procedure for estimation of parameters by the
method of inference functions (IFM) and obtaining asymptotic normality of these
estimators. Lastly, we demonstrate an application to count data for the number 
of people entering and exiting a building, and show how the different aspects 
of the model combine to produce a strong predictive model. We conclude by 
suggesting some further areas of application and by listing directions for 
future work.
}

\section{Introduction}

The abundance of data brought about by the digital revolution has increased the
availability of time series of counts. Such data appear in many areas, including
statistics, econometrics, and the social and physical sciences.  For independent
count data, generalised linear models \citep{McCullaghNelder1989}
are widely used. The most popular distribution is the Poisson distribution, which has
attractive properties and is in some respects the count analogue of the Gaussian
distribution.  One restrictive property of the Poisson distribution however is that the mean
and the variance are equal -- this is rarely observed in applications. Naturally, many
alternatives have been proposed, see \citet{CameronTrivedi2013} for a comprehensive
review. In particular, the most common departures from the Poisson distribution are models
based on the negative binomial distribution, hurdle models, zero-inflated models,
Poisson-Normal mixture models, and finite mixtures models. \citet{Fokianos2012}
considers integer-valued autoregressive models for count time series and discusses 
estimation for both the Poisson model and the negative binomial model. Whilst the 
negative binomial model can account for over-dispersion present in the data, 
we have yet to mention a fix for under-dispersed data. \citet{McShaneEtAl2008} 
developed a count model based on the Weibull distribution that can handle both 
under-dispersed and over-dispersed data. Building on this idea, 
\citet{KharratBoshnakov2018} extended this approach to create a rich and 
flexible family of renewal count distributions, which greatly extends the toolbox of
distributions available for modelling count data. 

While for independent data the focus is on the provision of suitable distributions, in time 
series modelling the dependence presents additional challenges.
Models developed for modelling the dynamics of (continuous) time series often
provide adequate results for count data. The classic examples are ARMA models
\citep{BoxJenkins1970} and their multivariate extensions, which can be dealt
efficiently with state space methods \citep{DurbinKoopman2012}. A fruitful
approach, employed in ARCH and GARCH models \citep{Engle1982,Bollerslev1986},
uses a separate equation to model directly the dependence of the variance on the
past.  In order to improve the predictive accuracy, the aforementioned models
have been augmented with additional exogenous covariates.  ARMAX models
\citep{Hannan1988armaxbook,LikothanassisDemiris1998} allowed covariates to be
added to processes following an ARMA model, while GARCH-X \citep{Engle2002}
added the same feature to GARCH models.  \citet{ShephardSheppard2010} introduced
HEAVY models to improve prediction in high-frequency data, while
\citet{HansenEtAl2012} developed the Realised GARCH model, a class of GARCH-X
models for returns with an integrated model for realized measures of volatility.
There have been many efforts to extend the continuous GARCH model to the
multivariate case, summarised by \citet{BauwensEtAl2006}. These fall into three
categories: direct generalisations of the univariate GARCH model (VEC, BEKK and
factor models), linear combinations of univariate GARCH models (generalised
orthogonal models and latent factor models), and nonlinear combinations of
univariate GARCH models (DCC, GDC and copula-GARCH models).

The above models do not make specific provision for the non-negativity and 
integer-valued nature of count data. One approach has been to use the generalised 
linear model (GLM) methodology for time series data with an appropriate distribution, 
see \citet{KedemFokianos2002} for more details. Another approach is to use a
thinning operator to imitate ARMA models. These models are called integer
autoregressive moving average (INARMA) models and details can be found in
\citet{Weiss2008}. Furthermore, an integer-valued analogue of the GARCH model
was proposed by \citet{FerlandEtAl2006}, called INGARCH, which uses Poisson
deviates rather than normal innovations. \citet{FokianosEtAl2009} also used the
GARCH model for inspiration, as they aspired to create a Poisson model for
integer-valued time series containing an autoregressive feedback mechanism
similar to the volatility in GARCH models. They called this model the Poisson
autoregressive model and later the properties were extended to negative binomial
autoregressive models by \citet{ChristouFokianos2014}. \citet{AgostoEtAl2016}
proposed a class of dynamic Poisson models allowing for additional
(exogenous) covariates to strengthen the predictions. This was referred to as
the Poisson autoregressive model with exogenous covariates (PARX).

All models for count data mentioned so far are univariate.
Whilst the Poisson distribution has been widely used for univariate count models,
multivariate generalisations have been relatively sparse so far. \citet{InouyeEtAl2017} provide 
a summary of multivariate (Poisson) distributions for count data, with methods including
multivariate extensions of a parametric (Poisson) distribution and copula modelling using
univariate (Poisson) marginal distributions. For example, \citet{Lui2012} formulates a 
bivariate Poisson integer-values GARCH (BINGARCH) model using the parametric
bivariate Poisson distribution and argues that, given a suitable multivariate Poisson 
distribution, his framework is capable of dealing with the multivariate case. For predicting
the scores of football matches, \citet{KoopmanLit2015} have applied a parametric 
bivariate Poisson model, \citet{McHaleScarf2011} have used Frank's copula 
with Poisson and negative binomial marginal distributions, and 
\citet{BoshnakovEtAl2017} have used Frank's copula with Weibull count distributions 
as marginal distributions.

Our interest in this article lies in the modelling of multivariate count data. We use a copula 
approach to extend the (univariate) PARX model of \citet{AgostoEtAl2016} to multivariate 
count time series. This approach is flexible and tractable. Use of covariates in the Poisson 
model offers clear potential for better modelling and by including the time series 
covariates we allow over-dispersed data to be considered by our model. Implementation 
in R \citep{R} is available in the developmental package PoARX \citep{RPoARX}.

This paper is organised as follows. Section~\ref{S:MultiPoARX} introduces the 
multivariate PoARX model and gives stationarity and ergodicity conditions. In 
Section~\ref{S:Estimation} we discuss estimation of parameters by the method of 
inference functions (IMF) and obtain asymptotic results for the resulting estimators. 
Next, we consider prediction in Section~\ref{S:Forecasting}, looking at the generating 
functions for future horizons. Then we demonstrate an application of the PoARX model
in Section \ref{S:Applications} by analysing a bivariate time series of count data from 
\citet{IhlerEtAl2006}.  The time series represent the number of people entering and 
exiting a building on the University of California, Irvine (UCI) campus. Exogenous 
covariates, such as the occurrence of a meeting or conference are included in the 
model to aid predictive accuracy. We summarise our findings in 
Section~\ref{S:Conclusion} and outline suggestions for future work. 

\section{The multivariate PoARX model}
\label{S:MultiPoARX}

In this section we present the new class of models, introducing the necessary
background material about the univariate PoARX model and copulas, before 
focusing on the two-dimensional case and generalising to higher dimensions.  
For the purpose of this article we focus on using Frank's copula to capture 
dependence between time series, but any suitable copula could be used.

\subsection{The univariate PoARX model}
 
First, a note on terminology -- \citet{AgostoEtAl2016} use the abbreviation PARX for
this model but we prefer PoARX since it seems to suggest more clearly ``Poisson''
and avoids confusion with other meanings of ``P'' in similar abbreviations. For
example, PAR is often used to mean periodic autoregression.

Let $\{Y_t; \ t = 1,2, \dots \}$ denote an observed time series of counts, so
that $Y_t \in \{0,1,2, \dots \}$ for all $t = 1,2, \ldots$.  Further, let
$x_{t-1} \in \mathbb{R}^r$ denote a vector of additional covariates considered
for inclusion in the model. We say that $\{ Y_t \}$ is a univariate
PoARX($p$,$q$) process and write $\{ Y_t \} \sim \PoARX_1(p,q)$, if its dynamics
can be written as follows:
\begin{equation}
\label{E:uniPoARXModel}
\begin{gathered}
Y_t \given \infoF[t-1] \sim \text{Poisson} (\lambda_t),\\
\lambda_{t} = \omega + \sum_{l=1}^p \alpha_l Y_{t-l} + 
  \sum_{l=1}^q \beta_l \lambda_{t-l} + \eta \cdot x_{t-1},
\end{gathered}
\end{equation}
where $\infoF[t-1]$ denotes the $\sigma$-field of past knowledge,
$\sigma \{ Y_{1-p}, \dots, Y_{t-1}, \lambda_{1-q}, \dots, \lambda_{t-1}, x_1,
\dots, x_{t-1} \} $, $\text{Poisson}(\lambda)$ denotes a Poisson distribution
with intensity parameter $\lambda$, $\omega \geq 0$ is an intercept term,
$\{ \alpha_1, \dots , \alpha_p \}$ and $\{ \beta_1, \dots , \beta_q \}$ are
non-negative autoregressive coefficients, and $\eta$ is a vector of non-negative
coefficients for the exogenous covariates. Thus, the model for the intensity,
$\lambda_{t}$, uses the past $p$ values of the process, the past $q$ values of
the intensity and the covariates. 

In order to ensure that the process is stationary and ergodic with polynomial moments of 
a given order, we place two further restrictions on the model \citep{AgostoEtAl2016}. 
Firstly, the autoregressive coefficients must obey the following condition,
\begin{equation}
\label{E:uniPoARXrestriction}
\sum_{i=1}^{\max\{p,q\}} (\alpha_i + \beta_i) < 1.
\end{equation}
Additionally, we require that each component of the exogenous covariates, denoted
$x_t(k)$ to avoid confusion later, follows a Markov structure, that is,
\begin{equation}
\label{E:Markov}
x_{t}(k) = g(x_{t-1}(k), \dots, x_{t-m}(k) ; \epsilon_t), 
 \qquad k = 1, \dots, r, \qquad
\end{equation}
for some $m>0$ and some function $g(\bm{x}, \epsilon)$ with vector $\bm{x}$
independent of the observed $Y_t$ and unobserved $\lambda_t$, and with 
$\epsilon_t$ an i.i.d. error term.

\subsection{Copulas}

Copulas provide a well-defined approach to model 
multivariate data, with the dependence structure considered 
separately from the univariate margins \citep{Joe2005}.
A copula, $C$, is a multivariate distribution function with all univariate
margins having the $U(0,1)$ distribution \citep{Joe1997}.  More specifically,
let $U_i \sim U(0,1)$ for $i = 1,\dots, \Kdim$, be uniformly distributed random
variables, not necessarily independent. Their joint distribution function is the
copula
\begin{equation*}
C(u_1, \dots u_\Kdim) = 
 \Pr \left( U_1 \leq u_1, \dots U_\Kdim \leq u_\Kdim \right), 
 \quad 
 0 \leq u_1, \dots, u_\Kdim \leq 1.
\end{equation*}
In particular, the copula $C$ is a function mapping the $\Kdim$-dimensional
unit cube, $[0,1]^\Kdim$, onto the interval $[0,1]$. Note that
the distribution corresponding to the copula is also called a copula.

The dependence structure for the random variables $U_1, \dots U_\Kdim$ is
contained in $C$, parametrised by a dependence parameter $\rho$, which
can be a vector.  Copula theory has developed from a theorem by \citet{Sklar},
which states that any multivariate distribution can be represented as a function
of its marginals.

\begin{theorem}[Sklar's Theorem]
Let $F$ be a joint distribution function with marginals 
$F_1, \dots F_\Kdim$. Then there exists a copula 
$C$:$[0,1]^\Kdim \to [0,1]$ such that
\begin{equation*}
F(y_1, \dots y_\Kdim) = 
 C \left( F_1(y_1), \dots F_\Kdim(y_\Kdim) \right), 
  \quad 
   y_1, \dots y_\Kdim \in \mathbb{R}.
\end{equation*}
\end{theorem}

Copulas allow for flexible joint modelling of multivariate data whilst
retaining control over the dependence structure between the
variables. Whilst the copula must act upon uniform random 
variables, it is straightforward to apply the probability integral 
transform \citep{PIT} to create the required variables. Furthermore, 
estimation of parameters of the univariate margins and the copula 
itself can be performed separately. This can be seen in the 
approach taken by \citet{Joe1997}, who suggested a two-stage 
process of estimation, fitting first the univariate margins to the 
respective variables before fitting the copula to find 
the dependence parameter.

An important class of copulas are called Archimedean copulas.  They are
developed using Laplace transforms and mixtures of powers of univariate
densities to create multivariate distributions.  They have many nice properties
and can be constructed easily \citep{Nelsen2006} from a generator function
$\varphi(\cdot)$ and its pseudo-inverse, $\varphi^{[-1]}(\cdot)$, defined as
follows.
\begin{definition}[Pseudo-Inverse]
Let $\varphi$ be a continuous, strictly decreasing function from 
$\mathbf{I} = [0,1]$ to $[0, \infty]$ such that $\varphi(1) = 0$. 
The pseudo-inverse of $\varphi$ is:
\begin{equation*}
  \varphi^{[-1]}(t)
  =
  \begin{cases}
    \varphi^{-1}(t) & 0 \leq t \leq \varphi(0)      , \\
    0              &  \varphi(0) \leq t \leq \infty .
  \end{cases}
\end{equation*}
\end{definition} 
The pseudo-inverse, $\varphi^{[-1]}$, is continuous and non-increasing on 
$[0, \infty]$ and strictly decreasing on $[0, \varphi(0)]$. If $\varphi(0) = \infty$, 
then $\varphi^{[-1]}(t)  = \varphi^{-1}(t)$. 

An Archimedean copula in $\Kdim$ dimensions is constructed by the following 
equation, given a generator function $\varphi(\cdot)$ \citep{Joe1997},
\begin{equation}
\label{E:ArchimedeanCopula}
C(u_1,\dots u_\Kdim) = \varphi^{[-1]} \left( \sum_{i=1}^\Kdim \varphi(u_i) \right)
.
\end{equation}
To ensure that this satisfies the conditions for a copula, see the conditions placed on
$\varphi(\cdot)$ and $\varphi^{[-1]}(\cdot)$ in \citet{McNeilNeslehova2009}.

Frank's copula \citep{Nelsen2006} is one example of an Archimedean copula where 
the dependence parameter can take any value except zero in the two-dimensional case
($\rho \in \mathbb{R} \backslash \{ 0 \}$). This is an advantage of Frank's copula over 
many other common Archimedean copulas, as we can account for both positive and 
negative dependence. The generator function is
\begin{equation}
\label{E:FCGenerator}
\varphi_\rho(t) = - \log \left( \frac{ \exp ( - \rho t) - 1}{
 \exp( - \rho) - 1} \right),
\end{equation}
and its pseudo-inverse can be written explicitly as
\begin{equation}
\label{E:FCGeneratorInverse}
  \varphi_\rho^{[-1]}(t) = \varphi_\rho^{-1}(t)
  =  - \frac{1}{\rho} \log \left( 1 + \exp(-t)(\exp(-\rho) - 1) \right)
\end{equation}
By substituting these functions into Equation~\eqref{E:ArchimedeanCopula}
we obtain Frank's copula. Since $\varphi_\rho(0) = \infty$, 
Equation~\eqref{E:FCGeneratorInverse} is true for all $t \ge 0$.
We use the subscript $\rho$ to distinguish Frank's
copula from the general case.

In higher dimensions, the dependence parameter is limited to values in $(0, \infty)$, 
but in any case the limit as $\rho \to 0$ corresponds to independence. Indeed, from 
the easily verifiable limits $\lim_{\rho \to 0} \varphi_\rho(t) = -\log(t)$ and 
$\lim_{\rho \to 0} \varphi^{-1}_\rho(t) = \exp(-t)$, it follows that
\begin{equation*}
\begin{aligned}
\lim_{\rho \to 0} C_\rho (u_1, \dots, u_\Kdim) 
 &= \exp \left( - \sum_{i=1}^\Kdim - \log(u_i) \right) \\
 &= \exp \left( \log \left( \prod_{i=1}^\Kdim u_i \right) \right ) \\
 &= \prod_{i=1}^\Kdim u_i
 ,
\end{aligned}
\end{equation*}
which is the joint cumulative density function of independent $U(0,1)$ random variables.

To conclude the discussion of copulas, we give the probability mass function (pmf) for 
$\Kdim$-dimensional discrete distributions \citep{Nelsen2006}. In the discrete case
the copula is no longer unique due to the presence of stepwise marginal distribution 
functions \citep{Joe2014}. Despite this issue, copula models are still valid constructions 
for discrete distributions \citep{GenestNeslehova2007}. The pmf is given as
\begin{equation}
\label{E:discreteCopulapmf}
\begin{aligned}
\Pr(Y_1 = y_1, \dots, Y_\Kdim = y_\Kdim)
 &= \sum_{l_1=0}^1 \cdots \sum_{l_\Kdim=0}^1
  (-1)^{l_1 + \dots + l_\Kdim} 
  \Pr(Y_1 \leq y_1 - l_1, \dots, Y_\Kdim \leq y_\Kdim - l_\Kdim) \\
 &= \sum_{l_1=0}^1 \cdots \sum_{l_\Kdim=0}^1
  (-1)^{l_1 + \dots + l_\Kdim} 
  C \left( F_1(y_1 - l_1), \dots, F_\Kdim(y_\Kdim - l_\Kdim) \right)
  ,
\end{aligned}
\end{equation}
where $C$ is any copula from Sklar's theorem. 

\subsection{The bivariate PoARX model}
\label{sect:BiPoARX}

We start with the two-dimensional case since it is of interest on its own and the notation 
is somewhat simpler.  Let $\{Y_t = (Y_t^1, Y_t^2), \ t = 1,2, \dots \}$ be a bivariate time 
series of counts with associated exogenous covariates 
$\{ x_{t-1}^j = (x_{t-1}^j(1), x_{t-1}^j(2))^\top, \ j = 1, 2 \}$. Then the collection of 
exogenous covariates associated with $Y_t$ is the matrix 
\begin{equation*}
x_{t-1} = (x_{t-1}^1, x_{t-1}^2)^\top = 
\begin{bmatrix}
x_{t-1}^1(1) & x_{t-1}^1(2) \\
x_{t-1}^2(1) & x_{t-1}^2(2)
\end{bmatrix}.
\end{equation*}
We say that $\{ Y_t \}$ is a bivariate PoARX($p$,$q$) process and write 
$\{ Y_t \} \sim \PoARXpq pq$, if each of the component time series is a univariate 
PoARX process (see Equation~\eqref{E:uniPoARXModel}) and the joint conditional 
distribution is a copula Poisson. 

More formally, let $\mathcal{D}(\lambda^{1},\lambda^{2};\rho)$ be a bivariate distribution 
based on Frank's copula with dependency parameter $\rho$ and marginals 
Poisson($\lambda^{1}$) and Poisson($\lambda^{2}$).  Let also $\{Y_{t}^{1}\}$ and 
$\{Y_{t}^{2}\}$ be univariate PoARX processes with intensities $\lambda_{t}^{j}$, for 
$j=1,2$. Letting $\lambda_{t} = \left( \lambda_{t}^{1}, \lambda_{t}^{2} \right)$,
denote by $\infoF[t-1]$ the $\sigma$-field generated by all past observations
and exogenous covariates:
\begin{equation*}
  \infoF[t-1]
  = \sigma \{ Y_{1-p}, \dots, Y_{t-1},
              \lambda_{1-q}, \dots, \lambda_{t-1},
              x_1, \dots, x_{t-1}
            \}
  .
\end{equation*}

The process $\{Y_t = (Y_t^1, Y_t^2), \ t = 1,2, \dots \}$ is a $\PoARXpq pq$ process if 
the conditional distribution of $Y_{t}$ is
\begin{equation*}
    Y_t \given \infoF[t-1] \sim \mathcal{D} 
     (\lambda_t^1, \lambda_t^2; \rho)
    ,
\end{equation*}
where $\lambda_{t}^{1},\lambda_{t}^{2}$ are the intensities of $\{Y_{t}^{1}\}$ and 
$\{Y_{t}^{2}\}$, respectively, with dynamics specified by the equations:
\begin{equation*}
\begin{gathered}
  Y_t^j \given \infoF[t-1] \sim \text{Poisson} (\lambda_t^j), 
   \qquad j = 1, 2; \\
  \lambda_{t}^j = \omega^j + \sum_{l=1}^p \alpha_l^j Y_{t-l}^j + 
    \sum_{l=1}^q \beta_l^j  \lambda_{t-l}^j + \eta^j \cdot x_{t-1}^j,
    \qquad j=1,2; 
\end{gathered}
\end{equation*}
where $\alpha^j_l, \beta^j_l \geq 0$ denote coefficients for the past values of the 
observations and intensities respectively, $\eta^j$ denotes the vector of (non-negative) 
coefficients for the exogenous covariates, and $\omega^j \geq 0$ denotes an (optional) 
intercept term.

From the above specifications it follows that the (bivariate) conditional distribution 
function of $Y_{t}$ is
\begin{equation*}
  F(y; \lambda, \rho)
  = C_{\rho} (F_1(y^1; \lambda^1), F_2(y^2; \lambda^2))
  ,
\end{equation*}
where $C_{\rho}$ is Frank's copula function, and $F_{1}$ and $F_{2}$ are
the distribution functions of the Poisson marginals, i.e.
\begin{equation*}
  F_j(x; \mu)
  = \sum_{k=0}^x \e{-\mu} \frac{\mu^k}{k!}, \qquad j = 1,2
  .
\end{equation*}

\subsection{The multivariate PoARX model}

The extension to the multivariate case is straightforward. Let 
$\{ Y_t = (Y_t^1, \dots, Y_t^\Kdim), \  t = 1,2, \dots \}$ be a multivariate time series and let 
$\{ x_{t-1}^j = (x_{t-1}^j(1), \dots, x_{t-1}^j(r))^\top, \ j = 1, 2, \dots, \Kdim \}$ 
be the matrix of exogenous covariates associated with $Y_t$.
We say that $\{ Y_t \}$ is a PoARX process and write $\{ Y_t \} \sim \PoARXpq[\Kdim]pq$, 
if each of the component time series is a univariate PoARX  process and the joint 
conditional distribution is a copula Poisson. Let the intensities of PoARX processes 
be $ \{ \lambda_{t}^{j}; \ t = 1, 2 \dots, \ j=1, \dots, \Kdim \}$ and be denoted using
$\lambda_t = \left( \lambda_{t}^{1}, \dots \lambda_{t}^{\Kdim} \right)$.

Analogously to the previous section, let $\mathcal{D}(\lambda^{1},\dots, \lambda^{\Kdim};\rho)$ be a 
multivariate distribution based on Frank's copula with marginal distributions 
Poisson($\lambda^{1}$), $\dots$, Poisson($\lambda^{\Kdim}$) and dependency parameter $\rho$.
Let also
\begin{equation}
\label{E:FC}
  C_{\rho} (u_1, \dots u_\Kdim)
  = \varphi_\rho^{-1} \left( \sum_{k=1}^\Kdim  \varphi_\rho(u_k) \right)
  ,
\end{equation}
where $\varphi_\rho$ and $\varphi_\rho^{-1}$ are the generator function and
its pseudo-inverse of the Frank's copula from Equations
\eqref{E:FCGenerator} -- \eqref{E:FCGeneratorInverse}.
Before stating the entire behaviour of the multivariate model, the distribution 
function corresponding to $\mathcal{D}(\lambda^{1},\dots, \lambda^{\Kdim};\rho)$ is
\begin{equation}
\label{E:multivariateCDF}
F(y; \lambda, \rho) = 
   C_{\rho} (F_1(y^1; \lambda^1), \dots, F_\Kdim(y^\Kdim; \lambda^\Kdim)).
\end{equation}

\begin{subequations} 
\label{E:PoARX}
  
The conditional distribution of $Y_{t}$ is a Frank's copula distribution
\begin{equation}
\label{E:PoARXa}
    Y_t \given \infoF[t-1] \sim \mathcal{D} (\lambda_t^1, \dots 
    \lambda_t^\Kdim; \rho)
    ,
\end{equation}
where $\infoF[t-1]$ denotes the $\sigma$-field defined by all previous 
observations and exogenous covariates,  
$\sigma \{ Y_{1-p}, \dots, Y_{t-1}, 
\lambda_{1-q}, \dots, \lambda_{t-1}, x_1, \dots, x_{t-1} \}$, where each term 
contains information on all components of the time series. As before, the 
dynamics of the components of $Y_{t}$ are specified by the equations:
\begin{gather}
Y_t^j \given \infoF[t-1] \sim \text{Poisson} (\lambda_t^j), \qquad 
j = 1, \dots, \Kdim;
\label{E:PoARXb}
\\
\lambda_{t}^j = \omega^j + \sum_{l=1}^p \alpha_l^j Y_{t-l}^j + \sum_{l=1}^q 
\beta_l^j  \lambda_{t-l}^j + \eta^j \cdot x_{t-1}^j, \qquad j=1, \dots, \Kdim;
\label{E:PoARXc}
\end{gather}
where $\alpha^j_l, \beta^j_l \geq 0$ denote coefficients for the past values of
the observations and intensities respectively, $\eta^j$ denotes the vector of 
(non-negative) coefficients for the exogenous covariates, and $\omega^j \geq 0$ 
denotes an (optional) intercept term. For each univariate process, the two conditions in 
Equations~\eqref{E:uniPoARXrestriction} and~\eqref{E:Markov} 
must hold. 
\end{subequations}

\subsection{Properties of multivariate PoARX}

Here we prove stationarity and ergodicity of PoARX models using the properties of 
univariate PoARX processes, developed in \citet{AgostoEtAl2016}, and 
$\tau$-weak dependence. $\tau$-weak dependence is a stability concept developed by 
\citet{DoukhanWintenberger2008} for Markov chains that implies stationarity and ergodicity.
To aid the establishment of asymptotic properties later, it is advantageous to express
each PoARX process in terms of a sequence of independent Poisson realisations. 
Specifically, introduce $\{N_t^j( \cdot ),t = 1, 2, \dots \}$ for $j = 1, 2, \dots, K$ and let 
each set be a sequence of independent Poisson processes of unit intensity, such
that $Y_t^j$ is equal to $N_{t}^{j}(\lambda_t^j)$, the number of events in the
time interval $[0, \lambda_t^j]$. Then the model can be rewritten as 
\begin{equation}
\label{E:multiPoARXProcess}
\begin{gathered}
Y_t^j = N_t^j(\lambda_t^j), \quad \text{for } j = 1,2, \dots, \Kdim, \\
\lambda_t^j = \omega^j + \sum_{l=1}^p \alpha_l^j Y_{t-l}^j + 
 \sum_{l=1}^q \beta_l^j \lambda_{t-l}^j + \eta^j \cdot x_{t-1}^j,
\end{gathered}
\end{equation}
assuming all terms used to initialise, 
$\{Y_0, Y_{-1}, \dots Y_{1-p}, \lambda_0, \lambda_{-1}, \dots \lambda_{1-q} \}$ are known 
and fixed, noting that each $\{ Y_t \}$ and $\{ \lambda_t \}$ is a $\Kdim$-dimensional 
vector. Now, we impose a simpler Markov structure to help state and prove the results,
\begin{equation}
\label{E:MarkovCovariates}
x_t^j(k) = g^j \left (x^j_{t-1}(k); \epsilon_t^j \right), 
  \qquad j = 1,\dots, \Kdim, \qquad k = 1, \dots, r.
\end{equation}
However, the statements hold for the more general structure found in 
Equation~\eqref{E:Markov}. We also make three assumptions
similar to those found in \citet{AgostoEtAl2016} for the univariate model.

\paragraph{Assumption 1 (Markov)} The innovations $\epsilon_t^j$ and Poisson 
processes $N_t^j(\cdot)$ are i.i.d. for all $j = 1, 2, \dots, \Kdim$.
\paragraph{Assumption 2 (Exogenous Stability)} 
$$\E \norm{g^j \left (x^j; \epsilon_t^j \right) - g^j \left (\tilde{x}^j; \epsilon_t^j \right)}^s
\leq \kappa \norm{x^j - \tilde{x}^j}^s$$ for some $\kappa < 1$ and 
$\E \norm{g^j \left( 0 ; \epsilon_t^j \right)}^s < \infty$ for all $j = 1, 2, \dots, \Kdim$, 
for some $s \geq 1$.
\paragraph{Assumption 3 (PoARX Stability)} 
$\sum_{i=1}^{\max (p,q)} \left( \alpha_i^j + \beta_i^j \right) < 1$, 
 for each $j = 1, 2, \dots, \Kdim$. \\

In the formulae below the operator $\vect$ has its usual meaning. For a matrix
$A$, $\vect(A)$ is a (column) vector obtained by stacking the columns of $A$ on
top of each other.  As a shorthand, $\vect(A_{1},\dots, A_{m})$ is equivalent to
the more verbose $\vect ( \vect(A_{1}), \dots, \vect(A_{m}) )$.

\begin{theorem}
\label{thm:StationaryPoARX}
Under Assumptions 1 -- 3 and the Markov assumption in 
Equation~\eqref{E:MarkovCovariates}, there exists a weakly dependent stationary 
and ergodic solution, $X_t^* = \vect \left( (Y_t^*, \lambda_t^*, x_{t-1}^*) \right)$,
to  Equations~\eqref{E:PoARX}. The solution is such that 
$\E \left( \norm{X_t^*}^s \right) < \infty$, where $s \geq 1$ is found in Assumption 2, 
$Y^*_t = (Y_t^{*1}, \dots Y_t^{*\Kdim})^\top$ and 
$\lambda^*_t = (\lambda_t^{*1}, \dots \lambda_t^{*\Kdim})^\top$
are $\Kdim$-vectors, and $x^*_{t-1} = (x_{t-1}^{*1}, \dots x_{t-1}^{*\Kdim})^\top$ is a
$\Kdim \times r$ matrix.
\end{theorem}
\begin{proof}
See \ref{A:ProofStationary}.
\end{proof}

A consequence of Theorem \ref{thm:StationaryPoARX} is that it allows PoARX models to
use the (weak) law of large numbers (LLN) for stationary and ergodic processes. To 
ensure the correct analysis of asymptotic behaviour, we need to be able to use the LLN
for any initialisation, rather than a set of fixed initial values. Lemma~\ref{L:PoARXLLN} 
extends the LLN to hold for this case. The proof is no different to the univariate case in 
\citet{AgostoEtAl2016}, where the reader is directed to \citet{KristensenRahbek2015}.

\begin{lemma}
\label{L:PoARXLLN}
Let $X_t = \vect \left( (Y_t, \lambda_t, x_{t-1})^\top \right)$ be a process satisfying 
$X_t = F(X_{t-1}; \xi_t)$ with $\xi_t$ i.i.d, 
$\E \norm{F(x; \xi_t) - F(\tilde{x}; \xi_t)}^s \leq \kappa \norm{x - \tilde{x}}^s$, 
and $\E \norm{F(0; \xi_t)}^s < \infty$. For any function $h(x)$ satisfying:
\begin{itemize}
   \item[(i).] $\norm{h(x)}^{1+\delta} \leq M(1 + \norm{x}^s)$ for some $M, \delta  > 0$,
   \item[(ii).] for some $c>0$ there exists $L_c > 0$ such that 
     $\norm{h(x) - h(\tilde{x})} \leq L_c \norm{x - \tilde{x}}$ for $\norm{x - \tilde{x}} < c$,
\end{itemize}
it holds that 
$$ \frac{1}{T} \sum_{t=1}^T h(X_t) \overset{P}{\to} \E \left( h(X_t^*) \right), \quad
 \text{as } T \to \infty. $$
\end{lemma}
\begin{proof}
See \citet{KristensenRahbek2015}, or apply the main result from 
\citet{LindnerSzimayer2005}.
\end{proof} 

\section{Estimation}
\label{S:Estimation}

Here we describe how the PoARX model can be estimated. We also provide asymptotic 
results for the estimated parameters.

We consider the model specified by Equations~\eqref{E:PoARX}, 
where we denote the unknown parameters by $\vartheta$. Then with 
$\alpha^j = \left( \alpha^j_1, \dots, \alpha^j_p \right)^\top$, 
$\beta^j = \left( \beta^j_1, \dots, \beta^j_q \right)^\top$, and 
$\eta^j = \left( \eta^j_1, \dots, \eta^j_r \right)^\top$,
\begin{equation*}
\begin{aligned}
 \vartheta &= \left( 
  \omega^1, (\alpha^1)^\top, (\beta^1)^\top, (\eta^1)^\top, 
  \dots, 
  \omega^\Kdim, (\alpha^\Kdim)^\top, 
   (\beta^\Kdim)^\top, (\eta^\Kdim)^\top, 
  \rho \right)^\top, \\
  &= \left( (\theta^1)^\top, \dots,  (\theta^\Kdim)^\top, \rho \right)
   ,
\end{aligned}
\end{equation*}
where $\theta^j \in \Theta^j \subset [0, \infty)^{1+p+q+r}$.

The probability mass function of the copula PoARX model, derived from the cumulative 
mass function as rectangle probabilities (compare to 
Equation~\eqref{E:discreteCopulapmf}), is 
\begin{multline*}
  \Pr(Y_t^1 = y_t^1, \dots, Y_t^\Kdim = y_t^\Kdim) \\
  = \sum_{l_1=0}^1 \cdots \sum_{l_\Kdim=0}^1
      (-1)^{l_1 + \dots + l_\Kdim} 
      C_\rho \left( F_1(y_t^1 - l_1; \lambda_t^1),
                    \dots, 
                    F_\Kdim(y_t^\Kdim - l_\Kdim, \lambda_t^\Kdim) 
                  \right),
\end{multline*}
with $C_\rho(\cdot)$ representing Frank's copula and 
\begin{equation*}
F_j(x; \mu) = \sum_{k=0}^{x} 
 \e{-\mu} \frac{\mu^{k}}{k!},
\qquad j = 1, \dots, \Kdim.
\end{equation*}
The conditional log-likelihood for $\vartheta$ given the multivariate observations 
$y_1, \dots, y_n$ with initial values $y_0$ and $\lambda_0$ 
(denoted by the $\sigma$-field $\infoF[0]$) is given by the following.
\begin{equation*}
\begin{aligned}
l(\vartheta) &= 
  \sum_{t=1}^n \log \left(
   \Pr((y_t^1, \dots y_t^\Kdim)^\top 
    \given\infoF[t-1]; \vartheta) \right) \\
  &= \sum_{t=1}^n l_t (\vartheta).
\end{aligned}
\end{equation*}
The maximum likelihood estimator (MLE) is
\begin{equation*}
\hat{\vartheta} = \argmax_{\vartheta \in \Theta} l(\vartheta).
\end{equation*}
However, with the large dimension of $\vartheta$ it is computationally more feasible to 
use a two-stage procedure known as the method of inference functions (IFM), 
developed by \citet{Joe2005}. The idea of IFM is to estimate the marginal parameters 
separately from the dependence parameter, hence reducing the dimension of the 
unknown parameters in each maximisation process. To perform this we need
the marginal log-likelihoods. When we consider the observations $y_1^j, \dots, y_n^j$ 
for each $j = 1, \dots, \Kdim$ separately, the marginal log-likelihood for $\theta^j$ can 
be written as
\begin{equation}
\label{E:MarginalLikelihood}
\begin{aligned}
l_j(\theta^j) &= 
  \sum_{t=1}^n \log \left(
   \Pr(y_t^j \given \infoF[t-1]; \theta^j) \right) \\
  &= -\lambda_t^j + y_t^j \log (\lambda_t^j) - \log (y_t^j!)
  ,
\end{aligned}
\end{equation}
with $\lambda_t^j$ calculated using Equation~\eqref{E:PoARXc}.

The IFM method is more explicitly stated as follows,
\begin{enumerate}
\item[(a)] the log-likelihoods $l_j(\cdot)$ of the $\Kdim$ univariate margins are 
  independently maximised to produce estimates 
  $\tilde{\theta}^1, \dots, \tilde{\theta}^\Kdim$ ;
\item[(b)] the function $l(\tilde{\theta}^1, \dots, \tilde{\theta}^\Kdim, \rho)$ 
  is maximised over $\rho$ to obtain $\tilde{\rho}$.
\end{enumerate}

Before we state the main result of this section we make a reference to the large sample 
properties of univariate PoARX obtained by \citet{AgostoEtAl2016}. In order to analyse 
these properties, conditions were imposed on the parameters and the exogenous 
covariates. 

\paragraph{Assumption 4} 
The space of possible parameters for each marginal distribution $j$, $\Theta^j$, 
is compact for all $j = 1,\dots, \Kdim$. This means that for all 
$\theta^j = ( \omega^j, \alpha^j, \beta^j, \eta^j ) \in \Theta^j$, 
$\beta^j_i \leq \beta_i^{j,U}$, for each $i = 1, \dots, q$, and 
$\omega^j \geq \omega^j_L$ for some constants $\omega^j_L > 0$ and 
$\beta_i^{j,U} > 0$ with $\sum_{i=1}^q \beta_i^{j,U} < 1$.
\paragraph{Assumption 5} 
The polynomials $A^j(z) := \sum_{i=1}^p \alpha^j_{0,i} z^i$ and 
$B^j(z) := 1 - \sum_{i=1}^q \beta^j_{0,i} z^i$ have no common roots; and for any 
$a \neq 0$ and $g \neq 0$, $\sum_{i=1}^p a_i Y^{*j}_{t-i} + \sum_{i=1}^r g_i x^{*j}_{i,t}$ 
has a non-degenerate distribution. This should be true for each $j = 1, \dots, \Kdim$. \\

Using Assumptions 1 -- 5 we can obtain consistency of the maximum likelihood 
estimators of the parameters for the $j^\text{th}$ univariate PoARX component
based on Equation~\eqref{E:MarginalLikelihood}. Equivalently, we can state that the 
IFM estimator (from part~(a) of the IFM procedure) of the multivariate PoARX model is 
consistent. Furthermore, if $\theta^j \in \interior \Theta^j$, then
\begin{equation*}
\sqrt{n} (\tilde{\theta}^j - \theta^j_0) \overset{d}{\to} 
 \mathcal{N} \left( 0, H^{-1}(\theta^j_0) \right), \qquad
 H(\theta^j) := - \E \left( 
 \frac{\partial^2 l_j^*(\theta^j)}{\partial \theta^j \partial (\theta^j)^\top} 
 \right),
\end{equation*}
where $l_j^* (\theta^j)$ denotes the marginal likelihood function evaluated at the stationary 
solution. The proof is equivalent to the proof of Theorem 2 in \citet{AgostoEtAl2016}.

Lastly, from the theory of inference functions \citep{Godambe1991, Joe2005}, we can 
deduce an asymptotic result for the IFM estimate of $\rho$,
\begin{equation*}
\sqrt{n} (\tilde{\rho} - \rho_0) \overset{d}{\to} 
 \mathcal{N} \left( 0, H^{-1}(\rho_0) \right), \qquad
 H(\rho) := - \E \left( 
 \frac{\partial^2 l^*}{\partial \rho \partial \rho^\top} 
 (\tilde{\theta}^1, \dots, \tilde{\theta}^\Kdim, \rho)
 \right).
\end{equation*}

We can now state our result about the asymptotic behaviour of the IMF
estimator of $\vartheta$, the full vector of parameters.
\begin{theorem}
\label{thm:CLTPoARX}
Suppose that Assumptions 1 -- 5 hold with $s \geq 2$ and the true value of
$\vartheta$ is denoted by $\vartheta_0$. Then $\vartheta$ is consistent and if 
$\vartheta \in \interior \Theta$,
\begin{equation}
\sqrt{n} (\tilde{\vartheta} - \vartheta_0) \overset{d}{\to} 
\mathcal{N} \left( 0, V \right),
\end{equation}
where details of asymptotic covariance matrix $V$ can be found in the 
proof.
\end{theorem}
\begin{proof}
See \ref{A:ProofCLT}.
\end{proof}

\section{Forecasting}
\label{S:Forecasting}

Forecasting with PoARX models is to some extent similar to the forecasting of GARCH-X
processes \citep{HansenEtAl2012}.  Predictions for the intensities can be obtained
recursively using Equation~\eqref{E:PoARXc} and the property
$\E (Y^j_t \given \infoF[t-1]) = \lambda^j_t$. This procedure also gives point
predictions for the process. However, there is substantial difference when predictive
distributions are required.

One-step ahead forecasts at time~$t$ of the intensities
$\lambda^j_{t+1}, \dots, \lambda^j_{t+h-1}$, given information~$\infoF[t]$, parameters
$\theta^j$, and covariates $x_t$ are:
\begin{equation}
\lambda^j_{t+1 \given t} = \omega^j + 
 \sum_{l=1}^p \alpha^j_l y^j_{t+1-l} + 
 \sum_{l=1}^q \beta^j_l \lambda^j_{t+1-l} + 
 \eta^j \cdot x^j_{t}, \qquad
 j = 1, \dots, K.
\end{equation}
By the specifications of the model, the one-step ahead marginal predictive
distributions are Poisson with predicted intensities computed above, i.e. for
each $j=1,\dots,K$, 
\begin{equation*}
P(Y^j_{t+1} = y \given \infoF[t]) = 
  \frac{\lambda^y  \exp (- \lambda ) }{y!}.
\end{equation*}
where $\lambda = \lambda^j_{t+1 \given t}$.
The joint predictive distribution is obtained by substituting the predicted
intensities in Equation~\eqref{E:multivariateCDF}.

For multi-step-ahead forecasts, the procedure is not so straightforward. Firstly, 
the computation of the $h$-step-ahead forecast at time~$t$ assumes that the
exogenous covariates $x_{t}, \dots, x_{t+h-1}$ are known. In practice, these 
will often need to be replaced by their own forecasts or projections. This is not a 
problem when the covariates are leading indicators, see the example in Section
\ref{S:Applications}. With a slight abuse of notation we use 
$\lambda^j_{t+h \given t}$ to represent the ``intensity for horizon $h$ conditional 
on $\infoF[t]$ and $x_{t}, \dots, x_{t+h-1}$''. We let this knowledge be denoted by 
the $\sigma$-field $\infoG[t]$. \citet{AgostoEtAl2016} assume that the predictive 
distribution for any horizon $h$ follows a Poisson distribution, 
$Y^j_{t+h \given t} \sim \text{Poisson} ( \lambda^j_{t+h \given t} ),$
and use it to obtain prediction intervals. However, we show below that the predictive 
distributions for $h \ge 2$ are not necessarily Poisson. Rather than compute 
the probabilities directly, we use an approach similar to \citet{Boshnakov2009} who 
derived predictive distributions (for a different class of models) using conditional 
characteristic functions. Since the Poisson distribution is discrete, it is more convenient 
to use probability generating functions. 

The probability generating functions can be calculated as follows, starting with
$h=2$. For a time series $Y_{t}$ following a PoARX process with intensity 
$\lambda_t$, we can write $\lambda_{t+2 \given t} = c_{t+2} + \alpha_{1}y_{t+1}$, 
where $c_{t+2}$ is measurable w.r.t. $\infoG[t]$.
In the derivation below we will need the following result:
\begin{align}
  \E(\exp \left( (-1 + z)\alpha_{1}y_{t+1} \right) \given \infoG[t] )
  &= \sum_{k=0}^{\infty} \frac{\lambda_{t+1}^{k}}{k!} 
     \exp \left(-\lambda_{t+1} \right)
     \exp \left( (-1+z)\alpha_{1}k \right) \nonumber
  \\ &=  \exp\left( -\lambda_{t+1} \right) \sum_{k=0}^{\infty}
                 \frac{(\lambda_{t+1}\e{(-1+z)\alpha_{1}})^{k}}{k!} \nonumber
  \\ &= \exp \left( -\lambda_{t+1} \right) \exp \left(
       \lambda_{t+1}\e{(-1+z)\alpha_{1}} \right) \nonumber
  \\ &= \exp \left( \lambda_{t+1}(-1 + \e{(-1+z)\alpha_{1}}) \right)
       \label{eq:aux1}
       .
\end{align}
The 2-step ahead forecast has the following generating function ($\pgfP[2]{z}$
depends also on $t$ but we omit that to keep the notation transparent):
\begin{align*}
  \pgfP[2]{z}
     &= \E(z^{Y_{t+2}} \given \infoG[t])
  \\ &= \E( \E(z^{Y_{t+2}} \given \infoG[t+1]) \given \infoG[t])
  \\ &= \E( \exp \left( (-1 + z)\lambda_{t+2} \right) \given \infoG[t] )
  \\ &= \exp \left((-1 + z)c_{t+2} \right) 
     \E(\exp \left( (-1 + z)\alpha_{1}y_{t+1} \right) \given \infoG[t] )
  \\ &=
       \begin{cases}
       \exp \left((-1 + z)c_{t+2} \right) & \text{if $\alpha_{1} = 0$,} \\
       \exp \left((-1 + z)c_{t+2} \right) \exp \left( \lambda_{t+1}(-1 + \exp{(-1+z)\alpha_{1}}) \right)
          & \text{if $\alpha_{1} \neq 0$, by \eqref{eq:aux1}.}
       \end{cases}
\end{align*}
We can see that if $\alpha_{1}\neq 0$, then $\pgfP[2]{z}$ is not Poisson, by the uniqueness 
property of generating functions. The joint distribution can be obtained by computing 
analogously the joint probability generating functions.

For $h>2$ the above calculation can be extended by repeatedly using the property of the
iterated conditional expectation. It can also be expressed recursively as follows:
\begin{align*}
  \pgfP[h]{z}
     &= \E(z^{Y_{t+h}} \given \infoG[t])
  \\ &= \E( \E(z^{Y_{t+h}} \given \infoG[t+1]) \given \infoG[t])
  \\ &= \E( \pgfP[h-1]{z} \given \infoG[t])
\end{align*}
Clearly, for $h \ge 2$ the forecast distribution is not necessarily Poisson. 
Nevertheless, we have that
\begin{lemma}
$\E( Y_{t+h} \given \infoG[t]) = \E( \lambda_{t+h} \given \infoG[t] ) = 
 : \lambda_{t+h \given t}$
\end{lemma}
\begin{proof}
For $h=1$, the claim follows from the specification of the model. For $h>1$ we
can use Equation~\eqref{E:PoARXc} and iterated conditional expectations to find that
\begin{align*}
  \E( Y_{t+h} \given \infoG[t])
     &= \E( \E(Y_{t+h} \given \infoG[t+h-1]) \given \infoG[t])
  \\ &= \E( \lambda_{t+h} \given \infoG[t])
  . \qedhere
\end{align*}
\end{proof}
Therefore, we can generate $h$-step ahead forecast of the intensity with the 
following equation,
\begin{equation}
\lambda_{t+h \given t} = \omega + 
 \sum_{l=1}^p \alpha_l Y_{t+h-l \given t} + 
 \sum_{l=1}^q \beta_l \lambda_{t+h-l \given t} + 
 \eta \cdot x_{t+h-1}.
\end{equation}
where
\begin{equation*}
Y_{t+k \given t} = 
  \left\{
	\begin{array}{ll}
		\lambda_{t+k \given t} & \mbox{if } k > 0, \\
		y_{t+k} & \mbox{if } k \le 0.
	\end{array}
 \right.
\end{equation*}
Prediction intervals can be obtained by computing the probabilities from the probability
generating functions discussed above. Since these are probably feasible only for small
horizons, simulation would be a more practical alternative. To obtain a prediction interval
for $Y_{t+h}^j$, simulate a trajectory of the PoARX time series until time $t+h$,
resulting in one simulated value $Y_{t+h}^j$. 
Repeating this process $B$ times allows access to the quantiles from which we can 
obtain a prediction interval for the time series. Simulating a joint predictive region is an 
area for further work and not discussed here.

\section{Applications}
\label{S:Applications}

We illustrate the use of PoARX models with a data set from
\citet{IhlerEtAl2006}, who used it in their work on event detection. The
computations were done with R \citep{R} using the implementation of the PoARX
models in package PoARX \citep{RPoARX}.

\subsection{Data}

The data contains counts of the estimated 
number of people that entered and exited a building over thirty-minute intervals of a UCI 
campus building. Counts were recorded by an optical sensor at the front door starting 
from the end of 23/07/2005 until the end of 05/11/2005. The data has periodic tendencies 
but is also influenced by events within the building causing an influx of traffic. Originally, 
the data was used to build a novel event detection framework under a Bayesian scheme. 
The counts of people going into (N$^{\text{I}}$(t)) and out of (N$^{\text{O}}$(t)) the 
building were both assumed to follow Poisson distributions and were used in a model 
to detect the occurrence of an event. Three weeks worth of the data in question is shown 
in Figure \ref{fig:Flows}. In total, there are 5040 observations, which corresponds to 
15 weeks of data.
\begin{figure}[ht]
\centering
\caption{Three weeks of counts for people entering and 
exiting a UCI campus building.}
\begin{subfigure}[b]{0.85\textwidth}
   \includegraphics[width=1\linewidth]{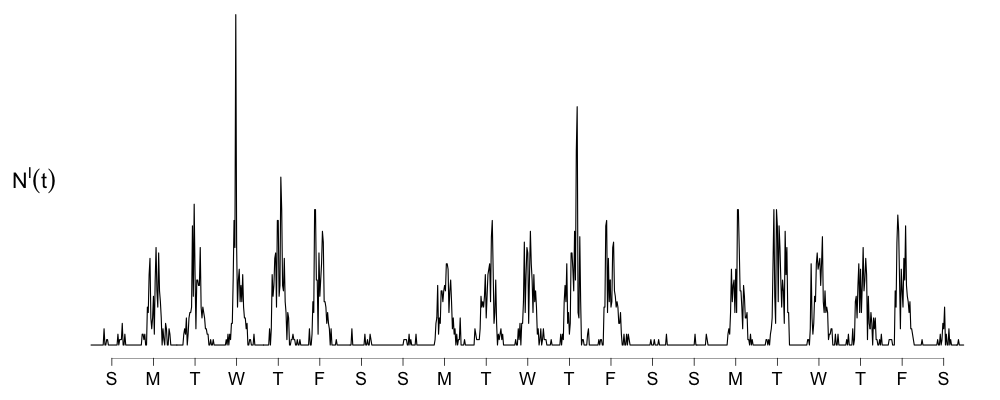}
   \caption{Entry data}
   \label{fig:FlowIn} 
\end{subfigure}
\begin{subfigure}[b]{0.85\textwidth}
   \includegraphics[width=1\linewidth]{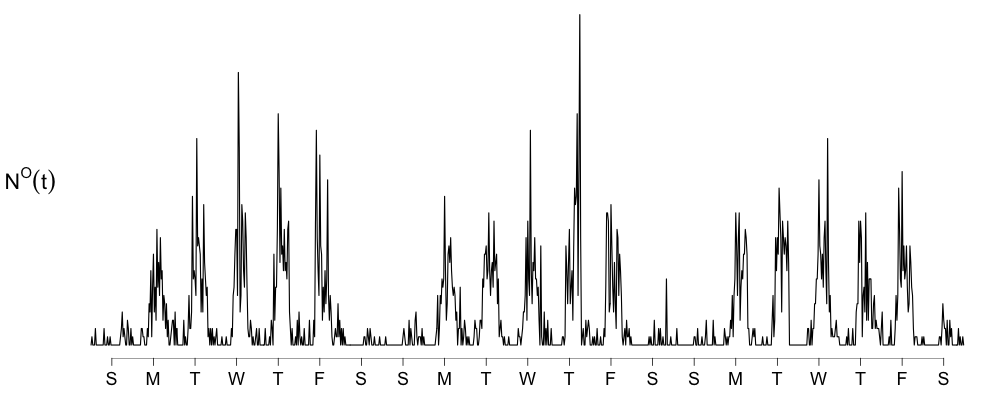}
   \caption{Exit data}
   \label{fig:FlowOut}
\end{subfigure}
\label{fig:Flows}
\end{figure}

In this application, we will estimate the number of people entering and exiting the building
using the Poisson distribution in the spirit of \citet{IhlerEtAl2006}. The basis of model 
predictions will be the lagged values of the observations and mean value, as well as some 
exogenous covariates. These covariates are all indicator variables, representing the 
following. The first is a ``weekday'' indicator, that takes value 1 when the day is 
Monday -- Friday. This corresponds to an uplift for working days. The second indicator is 
a ``daytime'' indicator, taking value 1 when the time is between 07:30 and 19:30, 
representing an uplift in the traffic during working hours. The third indicator is associated 
with the presence of an event occurring. For the flow count into the building, the variable 
takes the value 1 when an event will occur in the next hour. For the flow out of the 
building, the variable takes the value 1 in the hour after an event finished. These 
represent the arrival and departure of people coming to the building for the event. We 
will investigate whether the use of Frank's copula, hence the capturing of any positive 
or negative dependence, improves the prediction of the number of people entering and 
exiting the building.

\subsection{Estimation and in-sample model evaluation}

We fit four types of models to the data in an attempt to find the best predictive model. 
We first fit a model with no covariates - it uses only the time series aspects to predict 
upcoming counts. Model 1 uses this approach and treats the two counts independently, 
whereas model 2 fits the joint distribution of the flows using Frank's copula. We then add 
covariates to the models, seeking to improve the predictive accuracy of the two models. 
As mentioned, there are three covariates available for each time series. Model 3 uses the 
covariates along with the assumption of independence, whilst Model 4 uses Frank's 
copula with the covariates.

To assess the quality of our models, we used 5-fold cross validation \citep{Stone1974} on 
a training set to produce a cross-validated log score \citep{Bickel2007}. This was also 
the performance metric used to select the lagged values of the observations and means. 
Since we are modelling time series, we cannot leave out a fold that occurs in the middle 
of the data (thus disrupting the time series). Hence we choose overlapping folds, 
aggregating the log scores of predictions for each observation. Using the first 4000 
observations of the building data as a training set, we use 2000 observations in each 
fold of the cross-validation. The observations not used to estimate the model are used
for evaluation. The log score is calculated as follows. Let $r = (r_1, \dots, r_n)$ be a 
vector of probabilities for $i = 1, \dots, n$ observed events. Then the log score is
\begin{equation*}
L(r) = \sum_{i=1}^n \log (r_i).
\end{equation*}

For analysis, the lagged values chosen differed slightly for each time series. For the 
number of people entering the building (N$^{\text{I}}$(t)), we chose to use 4 lagged 
values for the observations (lags 1, 2, 48, 336) and 1 lagged value for the means (lag 1). 
Lagged values from the previous 2 observations represent the flow of people within the 
last hour, whilst the lag of 48 corresponds to the same time point on the previous day, 
and 336 to the same time point on the same day in the previous week. For the number 
of people exiting the building (N$^{\text{O}}$(t)) we used the same 4 lagged values for 
the observations (lags 1, 2, 48, 336) but included an extra lag for the mean values 
(lags 1, 48). These were chosen based on the cross-validated log scores. In 
Table~\ref{T:FittedModels} we present the values of the coefficients of the fitted models, 
where lags are sorted in increasing size (in other words $\alpha_3$ corresponds to the 
observations lag 48). The standard errors of parameters in Models~1 and 3 are of the 
order $10^{-4}$, and in Models~2 and 4 are of the order $10^{-5}$ or $10^{-6}$. 
This means that $\beta_{2}^{\text{O}}$ is not statistically significant in every model 
except Model~2, but when a new model is fitted without this variable we find that the 
strength of the predictions decreases. For this reason, we choose to keep the 
$48^\text{th}$ lagged mean in our models.

\begin{table}[ht]
\caption{Fitted models}
\centering
\begin{tabular}{| c | c c c c |}
\hline
Coefficient $\backslash$ Model & 1 & 2 & 3 & 4 \\[0.5ex] 
\hline
$\omega^{\text{I}}$ & 0.079 & 0.079 & 0.019 & 0.019 \\[5pt]
$\alpha_1^{\text{I}}$ & 0.390 & 0.390 & 0.396 & 0.396 \\[3pt]
$\alpha_{2}^{\text{I}}$ & 0.137 & 0.137 & 0.113 & 0.113 \\[3pt]
$\alpha_{3}^{\text{I}}$ & 0.054 & 0.054 & 0.048 & 0.048 \\[3pt]
$\alpha_{4}^{\text{I}}$ & 0.275 & 0.275 & 0.256 & 0.256 \\[5pt]
$\beta_{1}^{\text{I}}$ & 0.142  & 0.142 & 0.140 & 0.140 \\[5pt]
$\eta_{1}^{\text{I}}$ & - & - & 0.102 & 0.102 \\[3pt]
$\eta_{2}^{\text{I}}$ & - & - & 0.229 & 0.229 \\[3pt]
$\eta_{3}^{\text{I}}$ & - & - & 5.684 & 5.684 \\[5pt]
$\omega^{\text{O}}$ & 0.129 & 0.129 & 0.035 & 0.035 \\[5pt]
$\alpha_1^{\text{O}}$ & 0.347 & 0.347 & 0.342 & 0.342 \\[3pt]
$\alpha_{2}^{\text{O}}$ & 0.163 & 0.163 & 0.153 & 0.152 \\[3pt]
$\alpha_{3}^{\text{O}}$ & 0.049 &  0.049 & 0.045 & 0.045 \\[3pt]
$\alpha_{4}^{\text{O}}$ & 0.264 & 0.264 & 0.255 & 0.255 \\[5pt]
$\beta_{1}^{\text{O}}$ & 0.161 & 0.161 & 0.136 & 0.136 \\[3pt]
$\beta_{2}^{\text{O}}$ & 2.05e-04 & 2.05e-04 & 9.24e-10 & 9.24e-10 \\[5pt]
$\eta_{1}^{\text{O}}$ & - & - & 0.153 & 0.153 \\[3pt]
$\eta_{2}^{\text{O}}$ & - & - & 0.299 & 0.299 \\[3pt]
$\eta_{3}^{\text{O}}$ & - & - & 2.500 & 2.500 \\[5pt]
$\rho$ & - & 2.545 & - & 2.642 \\
\hline
\end{tabular}
\label{T:FittedModels}
\end{table}

In Table \ref{T:CVTrainScores} we present the cross-validated log score, AIC 
\citep{Akaike1974}, and BIC \citep{Schwarz1978} of the four models. Looking firstly at 
the information criteria, they both suggest that the best model is Model~4, which 
includes covariates and dependence.  Further, it seems that adding the covariates to 
the model improved the strength of both the model fitted with an independence 
assumption (Model~2 vs. Model~1) and the model using Frank's copula 
(Model~4 vs. Model~3). It also appears that the models using Frank's copula 
(Models 2 and~4) are better fits to the data than the independent case 
(Models 1 and~3, respectively).

However, we are interested in predictive accuracy, so we look mainly at the log scores. 
Firstly we notice that Model~2 appears to be the best model, while Model~1 is second. 
It seems as though the addition of the covariates weakens the fit of the model, despite 
the parameters of the relevant models being significantly greater than zero, statistically 
speaking. Furthermore, using this metric, we deduce that the use of Frank's copula 
improves the predictions compared to those using the independence assumption. The smallest
score and therefore the worst performance is found in the results from Model~3. This 
model contains covariates along with the independence assumption. However, since 
the two counts share common covariates, the assumption of independence is violated 
and we would speculate that this is the reason for the extreme score.

\begin{table}[ht]
\caption{Model training scores from cross-validated fit on 4000 observations}
\centering
\begin{tabular}{| c c c c |}
\hline
Model number & Log score & AIC & BIC \\[0.5ex] 
\hline 
1 & -15444 & 30252 & 30334 \\
2 & -15411 & 29802 & 29891 \\
3 & -25088 & 29800 & 29920 \\
4 & -16856 & 29269 & 29395 \\
\hline 
\end{tabular}
\label{T:CVTrainScores}
\end{table}

\subsection{Prediction and out-of-sample model evaluation}

As we are interested in the predictive strength of our model, it is a good idea to 
assess how the model performs predicting observations not in the original sample. Since
we only used the first 4000 observations in training, we can use the remaining 1040 
observations as a test set. Again using the log score to evaluate the performance, we 
display the results in Table \ref{T:TestScores}. 

\begin{table}[ht]
\caption{Model testing scores based on the 1040 out-of-sample observations}
\centering
\begin{tabular}{| c c |}
\hline
Model number & Log score \\[0.5ex] 
\hline 
1 & -4184 \\
2 & -4182 \\
3 & -4190 \\
4 & -4164 \\
\hline 
\end{tabular}
\label{T:TestScores}
\end{table}

From Table \ref{T:TestScores} we notice that Models~1-3 have similar scores, but 
Model~4 has a significantly lower log score. This would suggest that the combination 
of the time series aspects, the covariates and the multivariate modelling produces the 
most accurate out-of-sample predictions for this kind of data. Focusing on smaller 
comparisons, we first look at Models~1 and 2. There is a very small increase in 
performance by removing the independence assumption and using Frank's copula, but 
perhaps this is not worth the extra complexity gained from using a copula model. 
However between Models~3 and 4, the aforementioned increase in predictive performance
is evident, showing that when covariates are considered, the greater accuracy can be 
obtained using Frank's copula. Comparing Models~1 and 3 we see that there is a slight 
decline in predictive performance when the covariates are added. As mentioned earlier, 
one reason for this could be the violation of the assumption of independence due to the 
common covariates. However, between Models~2 and~4
the combination of covariates and copula produces the best performance. 

\section{Conclusion}
\label{S:Conclusion}

We introduced the multivariate PoARX model as an extension of the univariate PoARX 
model. Using previously established properties of the univariate PoARX model and 
copulas, we showed that our multivariate models inherit similar stability and large sample 
properties of the univariate case. We also established a law of large numbers.

For estimation of the parameters of multivariate PoARX models, we used the method of 
inference functions \citep{Joe2005}, which is computationally more efficient than the 
maximum likelihood method.  We established a central limit theorem for the 
parameters estimated by IFM. 

Our discussion of forecasting, especially predictive distributions for horizons
larger than one, seems novel even for the univariate PoARX models. In
particular, it is important to point out that the predictive distributions for
lags greater than one are not Poisson.

In the example in Section~\ref{S:Applications} we illustrated the use of bivariate 
PoARX models for modelling the counts of the number of people entering and exiting 
a building, using lagged values and covariates.  Overall, information criteria and 
out-of-sample prediction suggested that using both covariates and dependence 
parameters can provide better models. In this instance, we chose to use $k$-fold 
cross-validation coupled with the model assessment tool of the log score. However, 
this were relatively arbitrary choices, with no clearly defined methodology in place
for model assessment in general. Depending on the field of study, some people will
use information criteria, some will prefer scoring criteria. We feel that the analysis in 
Section~\ref{S:Applications} provides material for further thought and work on 
model evaluation for count data time series models.

We give here some examples of multivariate count data where multivariate PoARX models 
could be useful. The univariate PoARX model (or PARX model) has also been used to 
model the scores of a football match in \citet{AngeliniAngelis2017}. They used a 
univariate PoARX model for the goals scored by each team in the English Premier 
League and predicted the score coupling the processes independently. However, 
it has long been thought that there should be a dependence between teams competing 
in a match (see \citet{Maher1982} for the seminal paper in this area). 
Application of our multivariate PoARX model could be used to improve predictions for 
scores by considering such a dependence. Further applications could consider data 
modelled by a Poisson autoregressive process, and explore any influence of external 
factors. Such examples would be the Hyde Park Purse Snatchings and Presidential Vetoes 
from \citet{BrandtWilliams2000}, prices and times of trades made on the New York 
stock market from \citet{RydbergShephard2001} and the number of transactions per 
minute for the relevant stock from \citet{FokianosEtAl2009}.

There is also plenty of scope for further work. Our class of models uses Frank's
copula to jointly model Poisson marginal distributions. We did not have to use
Frank's copula -- if there is a belief that the dependence structure can be
captured in a different way, then other copulas can be used. Another direction
would be to consider distributions other than Poisson. We are considering the
possibility of using the renewal count distributions of
\citet{KharratBoshnakov2018}, mentioned in the introduction, which are
implemented in the R package Rcountr \citep{RCountr}.
Combining these renewal distributions with the ideas found in this paper
could lead to a fascinating new family of count time series models.
Additionally, exploring a time varying copula structure as seen in
\citet{KearneyPatton2000} may be advantageous in some applications.

\section{References}
\label{sec:references}

\bibliographystyle{abbrvnat}
\bibliography{PoARXBib}

\begin{appendix}

\section{Proof of Theorem \ref{thm:StationaryPoARX}}
\label{A:ProofStationary}

\begin{proof}
We start with the case $\rho = 0$ (independent time series). As each 
univariate time series satisfies the assumptions of Theorem 
\ref{thm:StationaryPoARX}, we know they are individually stationary and 
ergodic from \citet{AgostoEtAl2016}. Furthermore, the joint distribution is 
well defined as the product of each univariate probability. Hence the joint 
distribution is stationary. Lastly, for sets $A_1, \dots A_\Kdim  \in \mathbb{R}$, 
we have that
\begin{align*}
P &((Y_t^1, \dots, Y_t^\Kdim) \in (A_1, \dots A_\Kdim) \given 
\infoF[t-l]^1, \dots, \infoF[t-l]^\Kdim ) \\
 &= P (Y^1_t \in A_1 \given \infoF[t-l]^1 ) \cdot \dots 
  \cdot P (Y^\Kdim_t \in A_\Kdim \given \infoF[t-l]^\Kdim ).
\end{align*}
Using Theorem 1 from \citep{AgostoEtAl2016}, we have that 
\begin{equation*}
  P(Y^j_t \in B \given \infoF[t-l]^j ) \to P(Y^j_t \in B) 
  \ \ \text{as $l \to \infty$},
  \qquad \text{for $j = 1, \dots, \Kdim$}.
\end{equation*}
Hence, 
\begin{gather*}
P ((Y_t^1, \dots, Y_t^\Kdim) \in (A_1, \dots A_\Kdim) \given 
 \infoF[t-l]^1, \dots, \infoF[t-l]^\Kdim ) \to 
P ((Y_t^1, \dots, Y_t^\Kdim) \in (A_1, \dots A_\Kdim)) \\
 \text{as } l \to \infty, \text{ for any } A_1, \dots, 
  A_\Kdim \in\mathbb{R}.
\end{gather*}
This proves that independent PoARX processes are weakly dependent, 
therefore stationary and ergodic.

Now we move onto the case when $\rho \neq 0$. As before, we know that 
each time series in a multivariate PoARX model is stationary and ergodic.
Using similar arguments to \citet{MeitzSaikkonen2008}  we show the 
required joint result. Proving that the joint distribution is stationary is 
straightforward -- when $\rho \neq 0$, the cumulative mass function of 
the joint model is a simple, well-defined transformation of the univariate 
time series, as seen for the bivariate case in Equation~\eqref{E:jointCDF}. 
\begin{equation}
\label{E:jointCDF}
\begin{aligned}
F(y_t^1, y_t^2) &= \Pr(Y_t^1 \leq y_t^1, Y_t^2 \leq y_t^2)  \\
 &= -\frac{1}{\rho}  \log \left( 1 + \frac{ \left( 
  \exp \left( - \rho F_1(y_t^1) \right) - 1 \right) 
   \left( \exp \left( - \rho F_2(y_t^2) \right) - 1 \right) 
  }{\e{-\rho} - 1 } \right).
\end{aligned}
\end{equation}
To show the ergodicity, we must work harder. We show that the property of 
$\tau$-weak dependence holds for any number of dimensions using induction.

Start with $\Kdim=2$. Let 
\begin{equation*}
\begin{aligned}
\infoF[t-l]^1 = \sigma \left( Y^1_{t-l}, \lambda^1_{t-l}, 
 x^1_{t-l}, Y^1_{t-l-1}, \lambda^1_{t-l-1}, x^1_{t-l-1}, \dots \right) \\ 
\infoF[t-l]^2 = \sigma \left( Y^2_{t-l}, \lambda^2_{t-l}, 
 x^2_{t-l}, Y^2_{t-l-1}, \lambda^2_{t-l-1}, x^2_{t-l-1}, \dots \right)
\end{aligned}
\end{equation*}
and consider, for any sets $A, B \in \mathbb{R}$,
\begin{equation}
\label{E:2dConditionalProb}
\begin{aligned}
P &((Y_t^1, Y_t^2) \in (A, B) \given \infoF[t-l]^1, \infoF[t-l]^2 ) \\
&= P (Y^1_t \in A \given Y^2_t \in B, \infoF[t-l]^1, \infoF[t-l]^2 ) 
P (Y^2_t \in B \given \infoF[t-l]^1, \infoF[t-l]^2 ) \\
 &= P (Y^1_t \in A \given Y^2_t \in B, \infoF[t-l]^1 ) 
  P (Y^2_t \in B \given \infoF[t-l]^2 ).
\end{aligned}
\end{equation}
Using the definition of $\tau$-weak dependence inherited by univariate 
PoARX processes, 
\begin{equation*}
P(Y^2_t \in B \given \infoF[t-l]^2 ) \to 
 P(Y^2_t \in B) \text{ as } l \to \infty.
\end{equation*}
Using Equation~\eqref{E:jointCDF}, 
$P (Y^1_t \in A \given Y^2_t \in B, \infoF[t-l]^1)$ is a simple 
transformation of $P (Y^1_t \in A \given \infoF[t-l]^1)$. 
As $Y_t^1$ is a univariate PoARX process, 
\begin{equation*}
P(Y^1_t \in A \given \infoF[t-l]^1 ) \to P(Y^1_t \in A) 
 \text{ as } l \to \infty.
\end{equation*}
By applying the simple 
transformation for the conditional probability we find that 
\begin{equation*}
P (Y^1_t \in A \given Y^2_t \in B, \infoF[t-l]^1) \to 
 P (Y^1_t \in A \given Y^2_t \in B) \text{ as } l \to \infty.
 \end{equation*}
Thus, using Equation~\eqref{E:2dConditionalProb},
\begin{gather*}
P ((Y_t^1, Y_t^2) \in (A, B) \given \infoF[t-l]^1, 
 \infoF[t-l]^2 ) \to 
P ((Y_t^1, Y_t^2) \in (A, B)) \\
 \text{as } l \to \infty, \text{ for any } A,B.
\end{gather*}
This shows $\tau$-weak dependence, hence the bivariate PoARX copula model 
$(Y_t^1, Y_t^2)$ is stationary and ergodic.

Assume that this holds for $\Kdim=k$. Let 
$Y_t^{1:k} = (Y_t^1, \dots, Y_t^k)$. Then the assumption states that 
$Y_t^{1:k}$ is weakly dependent and hence ergodic.

Now we will prove for $\Kdim=k+1$. Let
\begin{equation*}
\begin{gathered}
\infoF[t-l]^j = \sigma \left( Y^j_{t-l}, \lambda^j_{t-l}, 
 x^j_{t-l}, Y^j_{t-l-1}, \lambda^j_{t-l-1}, x^j_{t-l-1}, \dots \right), 
  \quad j = 1, \dots, k, \\
\infoF[t-l]^{1:j} = \sigma \left( Y^{1:j}_{t-l}, 
 \lambda^{1:j}_{t-l}, x^{1:j}_{t-l}, Y^{1:j}_{t-l-1}, 
 \lambda^{1:j}_{t-l-1}, x^{1:j}_{t-l-1}, \dots \right), 
  \quad j = 2, \dots, k,
\end{gathered}
\end{equation*}
and for any sets $A \in \mathbb{R}$, and $B \in \mathbb{R}^k$, 
consider the following
\begin{align*}
P&((Y_t^{k+1}, Y_t^{1:k}) \in (A, B) \given \infoF[t-l]^{k+1}, 
 \infoF[t-l]^{1:k} ) \\
 &= P (Y^{k+1}_t \in A \given Y^{1:k}_t \in B, \infoF[t-l]^{k+1}, 
 \infoF[t-l]^{1:k} ) 
P (Y^{1:k}_t \in B \given \infoF[t-l]^{k+1}, 
 \infoF[t-l]^{1:k} ) \\
 &= P (Y^{k+1}_t \in A \given Y^{1:k}_t \in B, \infoF[t-l]^{k+1} ) 
  P (Y^{1:k}_t \in B \given \infoF[t-l]^{1:k} ).
\end{align*}
Because we know  $Y_t^{1:k}$ is weakly dependent from the assumption made, 
we have that
$$P(Y^{1:k}_t \in B \given \infoF[t-l]^{1:k} ) \to P(Y^{1:k}_t 
 \in B) \text{ as } l \to \infty.$$
$P (Y^{k+1}_t \in A \given Y^{1:k}_t \in B, \infoF[t-l]^1)$ can be 
thought of as a simple, well-defined transformation of 
$P (Y^{k+1}_t \in A \given \infoF[t-l]^{1:k})$. As $Y_t^{k+1}$ is a 
univariate PoARX process, 
\begin{equation*}
P(Y^{k+1}_t \in A \given \infoF[t-l]^{k+1} ) \to 
 P(Y^{k+1}_t \in A) \text{ as } l \to \infty,
\end{equation*}
 and as a result,
 \begin{equation*}
P (Y^{k+1}_t \in A \given Y^{1:k}_t \in B, \infoF[t-l]^{k+1}) \to 
 P (Y^{k+1}_t \in A \given Y^{1:k}_t \in B) \text{ as } l \to \infty
\end{equation*}
follows from the transformation. Thus,
\begin{gather*}
P ((Y_t^{k+1}, Y_t^{1:k}) \in (A, B) \given \infoF[t-l]^{k+1}, 
  \infoF[t-l]^{1:k} ) \to 
P ((Y_t^{k+1}, Y_t^{1:k}) \in (A, B)) \\
 \text{as } l \to \infty, \text{ for any } A \in \mathbb{R}, 
  B \in \mathbb{R}^k.
\end{gather*}
This shows that $Y_t^{1:(k+1)}$ is weakly dependent, hence ergodic, so the 
induction process holds.

We have now proven that the multivariate PoARX model, whether coupled 
independently or using Frank's copula, is jointly stationary and ergodic.
\end{proof}

\section{Proof of Theorem \ref{thm:CLTPoARX}}
\label{A:ProofCLT}

\begin{proof}
In the calculation of the IFM estimates $\vartheta$ we require the 
separate optimisations of $\Kdim$ marginal likelihoods. Each of 
these marginal likelihoods is a univariate PoARX process, and 
therefore under Assumptions 1-5 fulfils the requirements of 
Theorem 2 in \citet{AgostoEtAl2016}. Thus, for the parameters in 
$\theta^j$ for each $j = 1, \dots, \Kdim$,
\begin{equation*}
\sqrt{n} (\tilde{\theta}^j - \theta^j_0) \overset{d}{\to} 
 \mathcal{N} \left( 0, H_j^{-1}( \theta^j_0) \right), \qquad
 H_j(\theta^j) := - \E \left( 
 \frac{\partial^2 l_j^*(\theta^j)}{\partial \theta^j \partial (\theta^j)^\top} 
 \right).
\end{equation*}

First we consider the case of the PoARX models coupled independently, so there is no
dependence parameter to estimate. 
We should assume further here that there exists no 
condition that allows the observations to become dependent on each other.
Since any linear combination of the PoARX models must also follow a normal distribution,
we have the following result. Using 
$\theta = \vartheta_{(-\rho)} = (\theta^1, \dots, \theta^\Kdim)$ to denote the set of unknown 
parameters,
\begin{equation*}
\sqrt{n} (\tilde{\theta} - \theta_0) \overset{d}{\to}  \mathcal{N} \left( 0, V \right).
\end{equation*}
In this case, $V$ is a block diagonal matrix, where $H^{-1}_j(\theta_{j,0})$ are the non-zero 
entries. 
\begin{equation*}
V = \begin{bmatrix}
   H^{-1}_1(\theta^1_0) & 0 & \dots & 0 \\
   0 & H^{-1}_2(\theta^2_0) & \dots & 0 \\
   \vdots & \vdots & \ddots & \vdots \\
    0 & 0 & \cdots & H^{-1}_\Kdim(\theta^\Kdim_0)
\end{bmatrix}.
\end{equation*}

Now, in the case where Frank's copula is used to jointly model the PoARX models,
we require estimation of the $\rho$ using the profile log-likelihood with 
$\theta = \tilde{\theta}$. The regularity conditions for the theory of inference functions 
\citep{Godambe1991} hold for the dependence parameter, so we can use the asymptotic 
result,
\begin{equation*}
\sqrt{n} (\tilde{\rho} - \rho_0) \overset{d}{\to} 
 \mathcal{N} \left( 0, H_\rho^{-1}(\rho_0) \right), \qquad
 H_\rho(\rho) := - \E \left( 
 \frac{\partial^2 l^*}{\partial \rho \partial \rho^\top} 
 (\tilde{\theta}^1, \dots, \tilde{\theta}^\Kdim, \rho)
 \right).
\end{equation*}
Collecting all unknown parameters together, the theory of inference functions states 
that
\begin{equation*}
\sqrt{n} (\tilde{\vartheta} - \vartheta_0) \overset{d}{\to} 
 \mathcal{N} \left( 0, V \right),
\end{equation*}
for some asymptotic covariance matrix $V$. This matrix $V$ is given by
\begin{equation*}
V = (-D_g^{-1}) M_g (-D_g^{-1})^\top
\end{equation*}
where $M_g = \Cov (g(Y; \vartheta) )$ and 
 $D_g = \E \left(  \frac{\partial g(Y; \vartheta)}{\partial \vartheta^\top} \right)$ 
with $g = \left( \partial l_1 / \partial \theta_1, \dots \partial l_\Kdim /
\partial \theta_\Kdim, \partial l / \partial \rho \right)^\top$. 
Let $\mathcal{J}_{jk} = \Cov \left( g_j, g_k \right)$ be the covariance matrix between 
$g_j$ and $g_k$, 
and $\mathcal{I}_{jk} = - \E \left( \partial^2 l / \partial \theta^j \partial (\theta^j)^\top \right)$ 
for $1 \leq j,k \leq \Kdim$. This means that $\mathcal{I}_{jj} = H_j(\theta^j)$ is the 
Fisher information matrix for model. Lastly, we define
$\mathcal{I}_{mk} = - \E \left( \partial^2 l / \partial \theta^j \partial \rho \right)$ for 
$k = 1, \dots, \Kdim$. 
With this notation, the matrices can be partitioned as follows,
\begin{equation*}
- D_g = 
\begin{bmatrix}
    \mathcal{I}_{11} & 0 &  \dots & 0 & 0  \\
    0 & \mathcal{I}_{22} & \dots & 0 & 0  \\
    \vdots & \vdots & \ddots & \vdots & \vdots  \\
    0 & 0 & \dots & \mathcal{I}_{\Kdim \Kdim} & 0  \\
    \mathcal{I}_{m1} & \mathcal{I}_{m2}  & \dots 
      & \mathcal{I}_{m \Kdim} & \mathcal{I}_{mm}
\end{bmatrix},
\qquad
M_g = 
\begin{bmatrix}
    \mathcal{J}_{11} & \mathcal{J}_{12} &  \dots & 
      \mathcal{J}_{1\Kdim} & 0  \\
    \mathcal{J}_{21} & \mathcal{J}_{22} & \dots & 
      \mathcal{J}_{2\Kdim} & 0  \\
    \vdots & \vdots & \ddots & \vdots & \vdots  \\
    \mathcal{J}_{\Kdim 1} & \mathcal{J}_{\Kdim 2} & \dots & 
      \mathcal{J}_{\Kdim \Kdim} & 0  \\
    0 & 0 & \dots & 0 & \mathcal{J}_{mm}
\end{bmatrix}.
\end{equation*}
The only non-trivial calculations are $\Cov \left( g_j, g_d \right) = 0$ for
 $j = 1, \dots \Kdim$.  The proof of this can be found in the Appendix of 
 \citet{Joe2005}.

\end{proof}

\end{appendix}

\end{document}